\documentclass[11pt]{article}
\usepackage{epsfig,enumerate,amsmath,amsfonts,amssymb,amsthm,mathrsfs,lineno,ifpdf}
\usepackage{indentfirst,relsize}
\usepackage[numbers]{natbib}
\usepackage{setspace,graphicx}
\usepackage{latexsym}
\usepackage[all]{xy}
\usepackage[usenames,dvipsnames]{pstricks}
\usepackage{pst-grad} 
\usepackage{pst-plot} 
\usepackage[vlined,ruled,linesnumbered]{algorithm2e}

\usepackage{fullpage}

\theoremstyle{plain}
\newtheorem{theo}{Theorem}
\newtheorem{lemma}[theo]{Lemma}

\newtheorem{coro}[theo]{Corollary}

\newcounter{Cl}[theo]
\newtheorem{cl}[Cl]{Claim}

\theoremstyle{definition}
\newtheorem{defi}[theo]{Definition}
\newtheorem{remark}{Remark}
\newtheorem{obs}[theo]{Observation}

\newcommand{\dom}[3]{\gamma_{(#1,#2)}(#3)}

\newcommand{\prob}[1]{\mathrm{Pr}\left[#1\right]}
\newcommand{\bin}[2]{\mathrm{Bin}(#1,#2)}

\title{
$(1,j)$-set problem in graphs
}

\author{
Arijit Bishnu
\footnote{
Advanced Computing and Microelectronics Unit,
Indian Statistical Institute,
Kolkata, India
\texttt{arijit@isical.ac.in, paulsubhabrata@gmail.com}
}
\and
Kunal Dutta
\footnote{
D1: Algorithms \& Complexity,
Max-Planck-Institut f\"ur Informatik,
Saarbr\"ucken, Germany
\texttt{$\{$kdutta, agosh$\}$@mpi-inf.mpg.de}
}
\and
Arijit Ghosh
\footnotemark[2]
\and
Subhabrata Paul
\footnotemark[1]
}

\begin{document}


\maketitle

\begin{abstract}
A subset $D\subseteq V$ of a graph $G=(V,E)$ is a \emph{$(1,j)$-set} if every vertex $v\in V\setminus D$ 
is adjacent to at least $1$ but not more than $j$ vertices in $D$. The cardinality of a minimum $(1,j)
$-set of $G$, denoted as $\gamma_{(1,j)}(G)$,  is called the \emph{$(1,j)$-domination number} of $G$. 
Given a graph $G=(V,E)$ and an integer $k$, the decision version of the $(1,j)$-set problem is to decide 
whether $G$ has a $(1,j)$-set of cardinality at most $k$. In this paper, we first obtain an upper bound on 
$\gamma_{(1,j)}(G)$ using probabilistic methods, for bounded minimum and maximum degree graphs. 
Our bound is constructive, by the randomized algorithm of Moser and Tardos, 
We also show that the $(1,j)$-set problem is NP-complete for chordal graphs. Finally, 
we design two algorithms for finding $\gamma_{(1,j)}(G)$ of a tree and a split graph, for any fixed $j$, which 
answers an open question posed in \cite{12set}.

\end{abstract}

\section{Introduction}
\label{sec-introduction}

The concept of \emph{domination} and its variations is one of the most active area of research in graph 
theory because of its application in facility location problems, in problems involving finding a set of 
representatives, in monitoring communication or electrical networks, and in various other areas of practical 
applications (see \cite{Haynes1,Haynes2}). Over the years, many different variants of domination have been 
introduced and studied in the literature. The concept of $(i,j)$-set is a very interesting and recent 
variant of domination \cite{12set,YangWu}.

\subsection{Definitions}
\label{ssec-definitions}
For a natural number $m$, let $[m]$ denote the set $\{1,2,\ldots,m\}$.
Let $G=(V,E)$ be a graph. For $v \in V$, let $N_G(v)=\{u | uv \in E\}$ denote the open neighborhood of 
$v$ and $N_G[v]=N_G(v) \cup \{v\}$ denote the closed neighborhood of $v$. The degree of a vertex 
$v \in V$, denoted by $d_G(v)$, is the number of neighbors of $v$. Let $\Delta_G$ and $\delta_G$ denote the 
maximum and minimum degree of $G$. (We will remove the subscript $G$ where it is obvious from the context). 
Let $G[S]$ denote the subgraph induced by the vertex set $S$ on $G$. 
A \emph{tree} is a connected graph which has no cycle. A tree is called a \emph{rooted tree} if one of 
its vertices, say $r$, has been designated as the \emph{root}. The \emph{level} of a vertex is the number 
of edges along the unique path between it and the root. A set $S \subseteq V$ of a graph $G=(V,E)$ is an 
\emph{independent set} if no two vertices in $S$ are adjacent. If every pair of distinct vertices in 
$K\subseteq V$ are adjacent in $G$, then $K$ is called a \emph{clique}. A graph $G$ is \emph{chordal} if  
every cycle in $G$ of length at least four has a chord, that is, an edge between two non-consecutive 
vertices of the cycle. A graph $G=(V,E)$ is called a \emph{split graph} if $V$ can be partitioned into 
two sets, say $S$ and $K$, such that $S$ is an independent set and $K$ is a clique of $G$. Note that trees 
and split graphs are chordal graphs. A \emph{claw} is basically a $K_{1,3}$, a complete bipartite graph 
having one vertex in one partition and three vertices in the other partition. A vertex $u \in V$ is said 
to be dominated by a vertex $ v \in V$ if $u \in N_G[v]$. A set $D \subseteq V$ is called a \emph{dominating set} 
of $G$ if for every vertex $v \in V\setminus D$, $|N_G(v)\cap D|\geq 1$. The cardinality of a minimum dominating 
set of $G$ is called the \emph{domination number} of $G$ and is denoted by $\gamma(G)$. Note that, a dominating 
set $D$ dominates each vertex of $V\setminus D$ at least once. If, for some positive integer $i$, a dominating set 
$D_i$ dominates each vertex of $V\setminus D_i$ at least $i$ times, then $D_i$ is called a $i$-dominating set. 
A \emph{restrained dominating set} is a set $D_r\subseteq V$ where every vertex in $V\setminus S$ is adjacent to 
a vertex in $S$ as well as another vertex in $V\setminus S$. The cardinality of a minimum restrained dominating 
set of $G$ is called the \emph{restrained domination number} of $G$.

\subsection{Short review on $(i,j)$-set}
\label{ssec-short-review-1,j-set}

A set $D\subseteq V$ of a graph $G=(V,E)$ is called a \emph{$(i,j)$-set} if for every 
$v\in V\setminus D$, $i\leq |N_G(v)\cap D|\leq j$ for nonnegative integers $i$ and $j$, that is, every 
vertex $v\in V\setminus D$ is adjacent to at least $i$ but not more than $j$ vertices in $D$. The concept 
of $(i,j)$-set was introduced by Chellali et al. in \cite{12set}. Clearly, it is a generalization of the 
classical domination problem. Like domination problem, in this case, our goal is to find a $(i,j)$-set 
of minimum cardinality, which is called the \emph{$(i,j)$-domination number} of $G$ and is denoted by 
$\gamma_{(i,j)}(G)$. The decision version of $(i,j)$-set problem is defined as follows.

\vspace{10pt}

\noindent\underline{\textbf{$(i,j)$-Set problem ($(i,j)$-SET)}}
\begin{description}
  \item[Instance:] A graph $G=(V,E)$ and a positive integer $k\leq |V|$.
  \item[Question:] Does there exist a $(i,j)$-set $D$ of $G$ such that $|D|\leq k$?
\end{description}

In domination, we are interested in finding a set $D$ which dominates all the vertices of $V\setminus D$ 
at least once. But in some situation, we need to dominate each vertex at least $i$ times and at the same 
time, dominating a vertex more than $j$ times, might cause a problem. Basically, we are interested in 
finding a $i$-dominating set with a bounded redundancy. In these type of situations, we need the concept 
of $(i,j)$-set. Also, $(i,j)$-set is a more general concept which involves 
\emph{nearly perfect set} \cite{nearlyperfect}, \emph{perfect dominating set} \cite{perfect} (also known as 
\emph{$1$-fair dominating set} \cite{fair}) etc. as variants. There is a concept of \emph{set restricted domination} 
which is defined as follows: for each vertex $v\in V$, we assign a set $S_v$. A set $D_S$ is called a set restricted 
dominating set if for all $v\in V$, $|N_G[v]\cap D_S|\in S_v$. Note that if $S_v=[j]$ for all $v\in V$, 
we have a $(1,j)$-set. In that sense, $(1,j)$-set is a particular type of set restricted dominating set.

The concept of $(i,j)$-set has been introduced recently in 2013. Unlike other variations of domination, it 
has not been well studied until now. As per our knowledge, only two papers have 
appeared on $(i,j)$-set \cite{12set,YangWu}.
The main focus of \cite{12set} is on a particular $(i,j)$-set, namely $(1,2)$-set. In \cite{12set}, the 
authors have made a simple observation that for a simple graph $G$ with $n$ vertices, 
$\gamma(G)\leq \gamma_{(1,2)}(G)\leq n$. They have studied some graph classes for which these bounds are 
tight. They have shown that $\gamma(G)= \gamma_{(1,2)}(G)$ for claw-free graphs, $P_4$-free graphs, 
caterpillars etc. The authors have constructed a special type of split graph that achieves the upper bound. 
But there are some graph classes for which $\gamma_{(1,2)}(G)$ is strictly less than $n$. These graph classes 
involve graphs with maximum degree $4$, graphs having a $k$-clique whose vertices have degree either $k$ or 
$k+1$ etc \cite{12set}. In \cite{12set}, the authors have studied the $(1,3)$-set for grid graphs and showed 
that $\gamma(G)=\gamma_{(1,3)}(G)$. Using this result, they also showed that domination number is equal to 
restrained domination number. From complexity point of view, it is known that $(1,2)$-SET is NP-complete for 
bipartite graphs \cite{12set}. A list of open problems were posed in \cite{12set} indicating some research 
directions in this field. In \cite{YangWu}, the authors showed that some graphs with $\gamma_{(1,2)}(G)=n$ 
exist among some special families of graphs, such as planar graphs, bipartite graphs. These results answers 
some of the open problems posed in \cite{12set}. They also showed that for a tree $T$ with $k$ leaves, if 
$deg_G(v)\geq 4$ for any non-leaf vertex $v$, then $\gamma_{(1,2)}(T)=n-k$. Nordhaus-Gaddum-type inequalities 
are also established for $(1,2)$-set in \cite{YangWu}.

The main focus of \cite{12set} and \cite{YangWu} is $(1,2)$-set. In this paper, we study the more general 
set, namely $(1,j)$-set. Apart from the open problems mentioned in \cite{12set}, a bound on the 
$(1,j)$-domination number for general graphs is important. A general bound and its construction forms a 
major thrust of this paper, which is presented in Section $2$. In Section $3$, we tighten the hardness 
result by showing that $(1,j)$-set problem is NP-complete for chordal graphs. In Section $4$, we propose 
two polynomial time algorithms that calculate minimum $(1,j)$-domination number for trees and split graphs, 
that solves an open problem mentioned in \cite{12set}. Finally, Section $5$ concludes the paper.

\section{Upper bounds}
\label{sec-upper-bounds}

In this section, we shall prove an upper bound on the $(1,j)$-domination
number, i.e. $\gamma_{(1,j)}(G)$, of any graph $G=(V,E)$, having bounded minimum 
and maximum degree, for `sufficiently large' $j$.

In~\cite{AlonSpencer}, Alon and Spencer describe a similar 
upper bound on the domination 
number $\gamma(G)$, using probabilistic methods. Their strategy, (a classic 
example of the `alteration technique' in probabilistic methods), was to select 
a random subset $X$ of vertices as a partial dominating set, and then to 
include the set $Y$ of vertices not dominated by $X$, to get the final 
dominating set. However, such a strategy is \emph{a priori} not applicable for 
$(1,j)$-domination, because including or excluding vertices from the dominating 
set could change the number of dominating vertices adjacent to some vertex. 
Instead, we shall use a one-step process, and analyze it using the 
Lov\'{a}sz Local Lemma and Chernoff bounds to ensure that the conditions 
for $(1,j)$-set holds. 
Our proof also implies a randomized algorithm, using the Moser-Tardos constructive 
version of the Local Lemma~\cite{MoserTardos}, which would give a polynomial-time 
algorithm for obtaining a $(1,j)$-dominating set.

We first state two well-known results, Chernoff bound and Lov\'{a}sz Local Lemma, in a form suitable 
for our purposes. These results can be found in any standard text on 
probabilistic combinatorics, e.g.~\cite{AlonSpencer}. 

\begin{theo}[Chernoff bound]
Suppose $X$ is the sum of $n$ independent variables, each equal to $1$ with 
probability $p$ and 0 otherwise. Then for any $0 \leq \alpha$,
$$
	\prob{X > (1+\alpha) np} < \exp({-f(\alpha)np}),
$$
where 
$f(\alpha) = (1+\alpha)\ln(1+\alpha)-\alpha$.
\end{theo}

\begin{lemma}[Lov\'{a}sz local lemma] \label{lemmalovaszlocal}
    Let $\mathcal{A} = \{E_1,E_2,...,E_m\}$ be a collection of events 
    over a probability space such that each 
    $E_i$ is totally independent of all but the events in 
    $\mathcal{D}_i \subseteq \mathcal{A}\backslash \{E_i\}$. \\
    If there exists a real sequence $\{x_i\}_{i=1}^m$, $x_i \in [0,1)$, such 
    that 
\begin{eqnarray*}
    \forall i \in [m]\mbox{, } \prob{E_i} &\leq& x_i\prod_{j:E_j\in \mathcal{D}_i}
(1-x_j)\mbox{, then} \\
    \prob{\bigcap_{i=1}^m \overline{E}_i} &\geq& \prod_{i=1}^m (1-x_i) > 0. 
\end{eqnarray*} 
In particular, if for 
all $i$, $|\mathcal{D}_i| = d$ and $\prob{E_i} \leq p$, then, if $ep(d+1) \leq 1$,
then 
$$
	\prob{\bigcap_i \bar{E_i}} > \exp\left(-\frac{m}{d+1}\right).
$$
\end{lemma}

Before stating the main theorem, we need some definitions:
Given $\alpha \in \Re^+$, let 
$$
	f(\alpha) \stackrel{\rm def}{=} (1+\alpha)\ln (1+\alpha) -\alpha.
$$
Also let $s(\alpha) \stackrel{\rm def}{=} \min\{1,f(\alpha)\}$ and for $\Delta \in \mathbb{Z}^+$, let 
$$
	g(\Delta) \stackrel{\rm def}{=} \ln (2e(\Delta^2+1)) =
	1+\ln 2+ 2\ln \Delta + o_\Delta(1),
$$
where $e$ is the base of the natural logarithm.

\begin{theo}\label{thmubd}
  Given $j\in \mathbb{Z}^+$, let $\alpha >0$ be the maximum real number
such that 
$$ 
j+1 \geq (1+\alpha)\frac{\Gamma\, g(\Delta)}{s(\alpha)} \; \; \mbox{where} 
\; \; \Gamma = \frac{\Delta}{\delta}.
$$
Then, if such an $\alpha$ exists, 
$$\gamma_{(1,j)}(G) \leq (1+o_\Delta(1))\frac{g(\Delta)}{s(\alpha)\delta}n 
\leq (1+o_\Delta(1))\left(\frac{1+\ln 2 + 2\ln \Delta}{s(\alpha)\delta}\right)n \, .$$
Further, there is a randomized algorithm to obtain a $(1,j)$-dominating set of 
size at most $\frac{ng(\Delta)}{\delta s(\alpha)}$ that has expected runtime $O(n)$.
\end{theo}

\begin{proof}
Let $D \subset V$ be a subset of vertices obtained by tossing a coin for each vertex $v\in V$ 
independently and randomly with probability $p = \frac{g(\Delta)}{\delta s(\alpha)}$ and choosing 
$v$ if the coin comes up Heads. We shall show using the Local Lemma, that the subset $D$ is a 
$(1,j)$-dominating set with non-zero probability.

For each vertex $v \in V$, let $E_v$ be the event that $v$ is not $(1,j)$-dominated by 
$D$, i.e. that $v \not\in D$, and $|N(v) \cap D|\not\in [j]$. We need to show
that 
$$
	\prob{\bigcap_{v \in V}\bar{E_v}} > 0.
$$

In order to use the Local Lemma, consider the dependency graph formed 
by having the set of events $\{E_v\}_{v\in V}$ as vertices. Events $E_u$, $E_v$ $(u,v \in V)$
are dependent if and only if their outcomes depend on at least one common coin toss. Then clearly, the
events $E_u$, $E_v$ will be dependent if and only if 
$$
	N[u] \cap N[v] \neq \emptyset.
$$

This is possible only if the vertices $u$ and $v$ are at a distance at most 
$2$ from each other in the graph $G$. Hence, the degree of the dependency graph is at most 
$\Delta^2$. Now applying the symmetric form of the Local Lemma, we get that 
$$
	\prob{\bigcap_{v\in V} \bar{E_v}} > 0 \quad \mbox{if} \quad \prob{E_v} \leq \frac{1}{e (\Delta^{2}+1)}.
$$
We also 
need to bound the size of the selected subset $|D|$. However, this can easily be obtained
by applying a Chernoff bound to the output of the Local Lemma.
The proof of Theorem~\ref{thmubd} is therefore completed with the following 2 claims:
Let $X \stackrel{\rm def}{=} |D|$. With $p$, $\alpha$ and $E_v$ defined as above, 

\begin{cl} \label{clmlllreqmt}
   For all $v \in V$, 
   $$
	\prob{E_v} \leq \frac{1}{e(\Delta^2+1)}.
   $$
\end{cl}

\begin{cl} \label{clmchernoffoverall}
   For any $\varepsilon >\frac{1}{\sqrt{\Delta}}$, there exists $\Delta_0 \in \mathbb{Z}^+$, such 
   that for all $\Delta\geq \Delta_0$, we have
   $$\prob{ \Big( X < (1+\varepsilon)np \Big) \bigcap \left(\bigcap_{v\in V} \bar{E_v}\right)} > 0.$$
\end{cl}
Thus the set $D$ is of size at most 
$$
	(1+o_\Delta(1))n\left(\frac{1+\ln 2+2\ln \Delta}{\delta}\right),
$$ 
and every vertex in $V\setminus D$ has at least $1$ and at most $j$ neighbours in $D$. The bound on 
$\gamma_{(1,j)}(G)$ follows. 

We will elaborate on the randomized algorithm, via Moser-Tardos's local lemma implementation,
for obtaining such a $(1,j)$-dominating set in 
Remark~\ref{remark-2} at the end of this section.
%
\end{proof}

It only remains to prove the Claims~\ref{clmlllreqmt} and~\ref{clmchernoffoverall}.
\begin{proof}[Proof of Claim~\ref{clmlllreqmt}]
Given any vertex $v\in V$, define $X_v = |N(v)\cap D|$, and let $F_v$ denote the event that $X_v \not\in [1,j]$.
Then we have that 
\begin{eqnarray*}
   \prob{E_v} &=& \prob{E_v|F_v}.\prob{F_v} + \prob{E_v|\bar{F_v}}.\prob{\bar{F_v}} \\
   &=& (1-p)\prob{F_v} + 0 
\end{eqnarray*}

We shall prove the stronger condition : $\prob{\bigcap_{v\in V} \bar{F_v}} >0$.
Now, 
$$ \prob{F_v} = \prob{X_v < 1} + \prob{X_v > j} .$$
Observe that $X_v$ has the binomial distribution $\bin{d(v)}{p}$.

Note that, if $j > d(v)$, then the event $F_v$ occurs only when $X_v = 0$ and for $j \leq d(v)$, the event 
$F_v$ can occur when $X_v=0$ or when $X_v > j$. Therefore, 
\begin{equation}
   \prob{F_v} = \left\{ \begin{array}{ll}
                        (1-p)^{d(v)} & \mbox{ if } j > d(v) \\
                        (1-p)^{d(v)} + \prob{X_v > j} & \mbox{ if } j \leq d(v)
                        \end{array} \right. 
\end{equation}

By the premise of the Theorem,
we get that 
$$ j+1 \geq (1+\alpha)\frac{\Gamma g(\Delta)}{s(\alpha)} \geq (1+ \alpha)\frac{g(\Delta)}{s(\alpha)},$$
and hence, substituting the value of $p$, we get
$$ 
	f(\alpha)d(v)p \geq \frac{f(\alpha)\, g(\Delta)}{s(\alpha)} \geq g(\Delta),
$$
since $f(\alpha) \geq s(\alpha)$.
Substituting in the expression for $f(v)$ gives
$$ (1-p)^{d(v)} = e^{d(v)\ln (1-p)} \leq e^{-d(v)p} \leq \frac{1}{2e(\Delta^2+1)},$$
since $d(v) \geq \delta$.
To compute $\prob{X_v \geq j+1}$, we use the Chernoff bound:
\begin{eqnarray*}
   \prob{X_v \geq j+1} &\leq& \prob{\bin{d(v)}{p} \geq j+1} \\
   &\leq& \prob{\bin{d(v)}{p} \geq (1+\alpha)d(v)p} \\
   &\leq& \exp(-f(\alpha)d(v)p) \\
   &\leq& \frac{1}{2e(\Delta^2+1)}
\end{eqnarray*}
where the last inequality follows from the choice of $p$.

Therefore, we get that
\begin{eqnarray*}
  \prob{F_v} = \prob{X_v < 1} + \prob{X_v > 1}
  \leq \frac{2}{2e(\Delta^2+1)} 
\end{eqnarray*}
and hence that $\prob{F_v} \leq \frac{1}{e(\Delta^2+1)}$.
\end{proof}

\begin{proof}[Proof of Claim~\ref{clmchernoffoverall}]
To show that $\prob{A \cap B} > 0$ where $A,B$ are events in a probability space,
it suffices to 
show that
$$ 
	\prob{\bar{A} \cup \bar{B}} \leq \prob{\bar{A}} + \prob{\bar{B}} <1,\; \mbox{i.e.,}\;
	\prob{B}-\prob{\bar{A}} >0.
$$
Taking $A$ to be the event $(X < (1+\varepsilon)np)$ and $B$ to be $\left(\bigcap_{v\in V}\bar{E_v}\right)$, we 
shall first upper bound $\prob{A}$, and then use the lower bound on 
$\prob{B}$ from the Local Lemma.
Using Chernoff bound, we get that
\begin{eqnarray*}
 \prob{X \geq (1+\varepsilon)np} 
 &\leq& \exp{\left(-\frac{\varepsilon^2np}{3}\right)} \\
 &\leq& \exp{\left(-\frac{\varepsilon^2ng(\Delta)}{3\delta}\right)} 
\end{eqnarray*}
Now, from the Local Lemma, we get that 
\begin{eqnarray*}
	\prob{\bigcap_{v\in V}\bar{E_v}} &>& \left(1- \frac{1}{\Delta^2+1}\right)^{n} \\ 
	&\approx& \exp\left(-\frac{n}{\Delta^2+1}\right).
\end{eqnarray*}
Let $\varepsilon = \frac{\sqrt{c\delta}}{\Delta}$, where $c>0$ is any constant. Then we get that for sufficiently 
large $\Delta$, 
\begin{eqnarray*}
  \prob{B}- \prob{\bar{A}} &\geq& \exp\left(-\frac{n}{\Delta^2+1}\right) - \exp\left(\frac{-\varepsilon^2ng(\Delta)}{3\delta}\right) \\
   &>& 0
\end{eqnarray*}
since for $\varepsilon = \frac{\sqrt{c\delta}}{\Delta}$, we get that 
$$
	\frac{\varepsilon^2ng(\Delta)}{3\delta} \geq \frac{cn\ln\Delta}{3\Delta^2} > \frac{n}{\Delta^2+1}.
$$
\end{proof}
In particular, taking $\alpha=e-1$ and $G$ $d$-regular, we get: 
\begin{coro}\label{cor-1,j-dom-d-regular}
 If $G$ is a $d$-regular graph, and $j > eg(d)$ then 
 $$
	\dom{1}{j}{G} \leq (2+o_{d}(1))\frac{n \ln d}{d}.
 $$
\end{coro}

\begin{remark}\label{remark-1}
We remark that from the known lower bounds on the domination number of random graph, our results 
can be seen to be tight up to constant multiplicative factors, since $\gamma_{(1,j)}(G) \geq \gamma(G)$. 
For instance, the result of Glebov, Liebenau and Szab\`{o}~\cite{GlebovLiebenauSzabo}, implies 
that there exist graphs $G$ on $n$ vertices such that their domination number 
$\gamma(G) \geq \frac{n\, \log d}{d}$.
\end{remark}

\begin{remark}\label{remark-2}
 Elaboration on the Moser-Tardos's (MT) implementation: 
 We set up a SAT formula for domination, where each vertex $v_i \in V$ corresponds to a variable $x_i$ and 
 there is a clause $C(v)$ corresponding to each neighbourhood $N(v)$. $x_i= \mbox{``true''}$ means the vertex $v_i$ is 
 selected in the dominating set. Clause $C(v)$ is said to have failed if it is not satisfied 
 in the given assignment. In addition, for each vertex $v \in V$ having degree $d(v) > j$, there is a unique clause for
 every subset of $N(v)$ of size $j+1$, which fails only if all the vertices in the corresponding subset of $N(v)$ are 
 selected in the dominating set. Now we run the MT algorithm on this formula, 
 (i.e. take a random assignment where 
 each variable is set to true independently with the probability $p$ used in the proof; choose an arbitrary 
 failing clause and randomly reset all variables inside the clause; repeat until all clauses are satisfied). 
 The LLL condition guarantees that the MT algorithm will terminate in expected time linear in the number of clauses, i.e. $O(n\Delta^{j+1}) 
 = O(n^{j+2})$, and when this happens, the Chernoff bound guarantees that with high probability, not more 
 than $(1+o_\Delta(1))\frac{g(\Delta)n}{s(\alpha)\delta}$ many variables will be set to $\mbox{``true''}$.
\end{remark}

\section{NP-complete for chordal graphs}
\label{sec-NP-complete-chordal-graphs}

In this section, we show that $(1,j)$-SET is NP-complete when restricted to chordal graphs. Note that, 
for $j=1$, the problem is basically perfect domination problem, which is known to be NP-complete for 
chordal graphs \cite{perfectchordal}. For $j\geq 2$, we prove the NP-completeness by using a reduction 
from Exact $3$-Cover problem (EX$3$C), which is known to be NP-complete \cite{garey}. 

\vspace{10pt}

\noindent\underline{\textbf{Exact $3$-Cover problem (EX$3$C)}}
\begin{description}
  \item[Instance:] A finite set $X$ with $|X|=3q$, where $q$ is a positive integer and a collection $C$ of $3$-element subsets of $X$.
  \item[Question:] Is there a subcollection $C'$ of $C$ such that every element of $X$ appears in exactly one element of $C'$?
\end{description}

\begin{theo}\label{thm-NP-harness-chordal-graphs}
	$(1,j)$-SET is NP-complete for chordal graphs.
\end{theo}

Clearly $(1,j)$-SET for chordal graph is in NP. We describe a polynomial reduction from 
EX$3$C to $(1,j)$-SET for chordal graphs. Given any instance $(X,C)$ of EX$3$C, we obtain a 
chordal graph $G=(V,E)$ and an integer $k$ such that EX$3$C has a solution if and only if $G$ has 
a $(1,j)$-set of cardinality at most $k$.

Let $X=\{x_1, x_2, \ldots, x_{3q}\}$ and $C=\{C_1, C_2, \ldots, C_t\}$ be an arbitrary instance of EX$3$C. 
The vertex set of the newly formed graph $G=(V,E)$ is
formed as a disjoint union of $V_1, V_2$ and $V_3$, that is, 
$$
	V = V_1 \sqcup V_2 \sqcup V_3.
$$
For each $C_p\in C$, 
$p\in [t]$, we have a claw centered at a vertex $u_p$ and let $v_p, y_p$ and $z_p$ be the pendant 
vertices of that claw. The set $V_1$ is given by 
$$
	V_1=\bigcup_{p=1}^t \left\{u_p, v_p, y_p, z_p \right\}.
$$ 
Also, we have a set 
$V_2$ of $3q$ vertices $x_1, x_2, \ldots, x_{3q}$, each corresponding to an element of $X$. Furthermore, for 
each $i\in \{1,\, 2, \ldots,\, 3q\}$, we add a gadget $G_i$, as shown in Figure~\ref{fig1}. The gadget $G_i$ is 
basically a forest of $q$ number of rooted trees of depth 2 rooted at the vertices $w_1^i, w_2^i, \ldots w_{q}^i$ 
as shown in Figure~\ref{fig2}. In the gadget $G_i$, each $w_1^i,\, w_2^i,\, \ldots,\, w_{q}^i$ has $j$ children and 
each of these $j$ children has $2$ more children. The set $V_3$ is given by 
$$
	V_3=\bigcup_{i=1}^{3q} V(G_i),
$$ 
where 
$V(G_i)$ is the vertex set of $G_i$. Now we add the edges between $x_i$ and $v_p$ if the element corresponding 
to $x_i$ is in $C_p$. Note that degree of each $v_p$ is $4$ for all $p\in \{1,\, 2,\, \cdots,\, t\}$. Also, we add 
edges between every pair of distinct vertices of $\{x_1,\, x_2,\, \ldots,\, x_{3q}\}$, making it a clique. Finally, 
for each $i\in \{1, 2, \ldots, 3q\}$, we add the edges $x_iw_1^i,\, x_iw_2^i,\, \ldots,\, x_iw_q^i$. The construction 
of $G$ from the instance $(X,C)$ of EX$3$C is illustrated in Figure~\ref{fig1}. Clearly, the graph $G$ is a 
chordal graph. Let $k=t+q+3jq^2$.

Theorem~\ref{thm-NP-harness-chordal-graphs} directly follows from the following result. 

\begin{lemma} 
	EX$3$C has a solution if and only if $G$ has a $(1,j)$-set of cardinality at most $k=t+q+3jq^2$.
\end{lemma}

\begin{proof}
Suppose the instance $(X,C)$ has a solution $C'$. Since each 
element of $X$ is covered by exactly one element of $C'$, $|C'|=q$. For each gadget 
$G_i$, $i\in [3q]$, let $S_i$ be the set of all children of $w_1^i,\, w_2^i,\, \ldots,\, w_{q}^i$. 
Clearly, each $S_i$ contains $jq$ vertices. We form a set $D$ as follows: 
$$
	D = \left\{u_i|~1\leq i\leq t\right\} \bigcup \left\{v_p|~C_p\in C' \right\} \bigcup 
	\left(\bigcup_{l=1}^{3q} S_l \right). 
$$
Since $|C'|=q$, $D$ contains $t+q+3jq^2$ 
vertices. One can easily check that $D$ forms a $(1,j)$-set of $G$.

\begin{figure}
\begin{center}
  \includegraphics[width = 16cm]{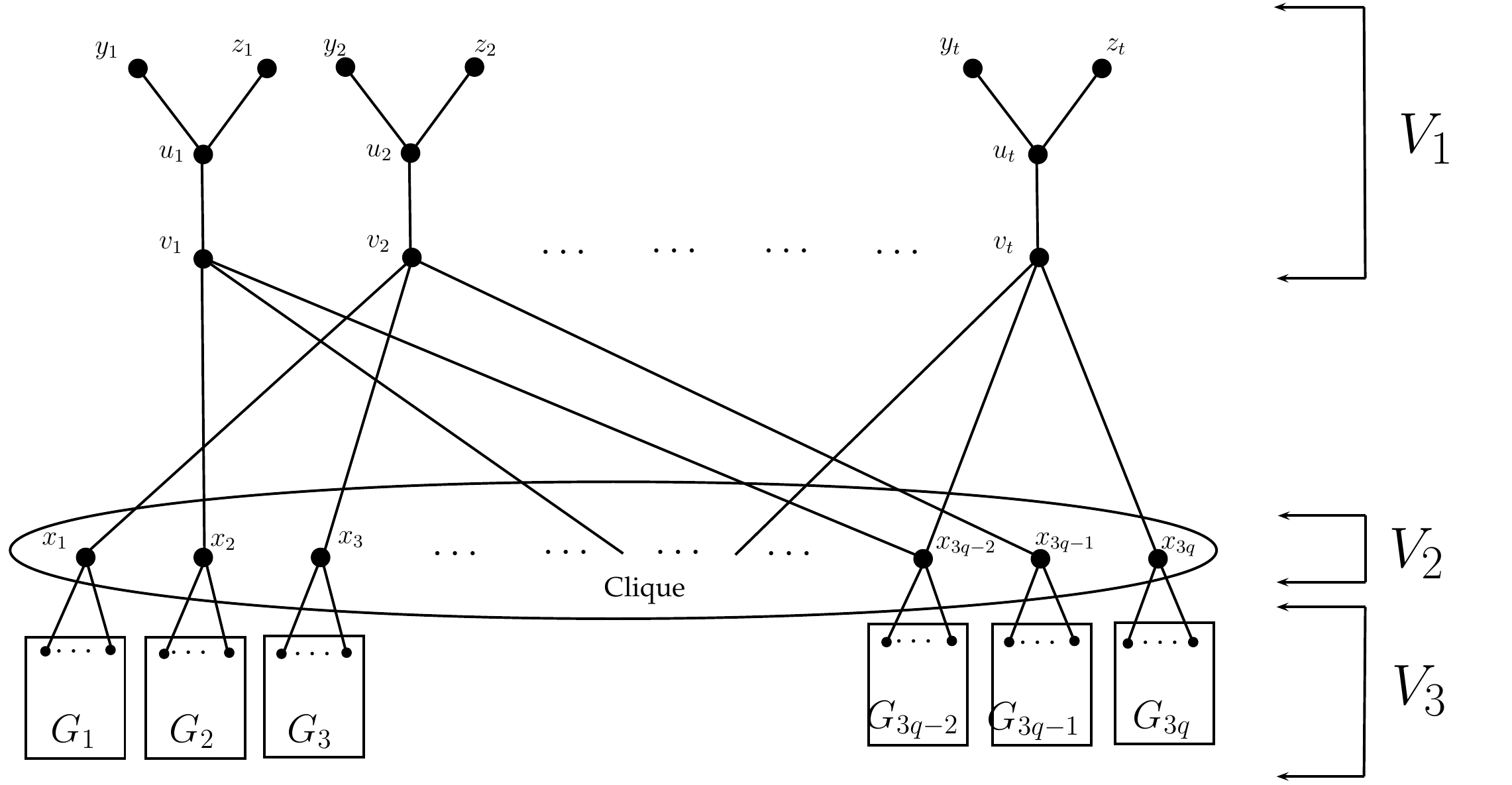}
\end{center}
\caption{Reduction from EX$3$C to $(1,j)$-SET}
\label{fig1}
\end{figure}

\begin{figure}
\begin{center}
  \includegraphics[width = 13cm]{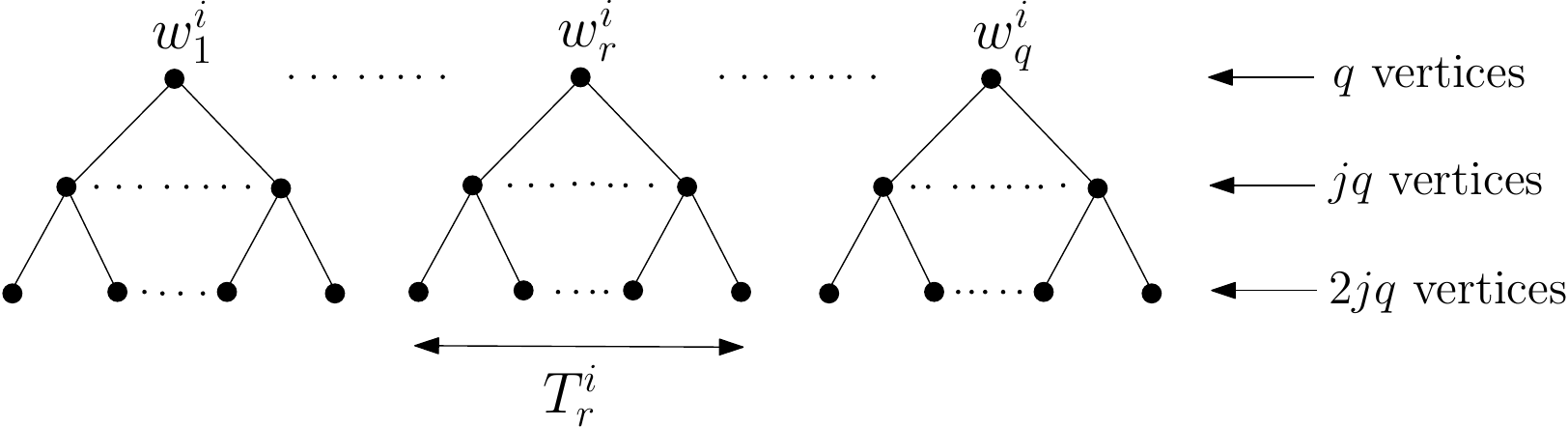}
\end{center}
\caption{Gadget $G_i$ corresponding to $x_i$}
\label{fig2}
\end{figure}

Conversely, suppose that $G$ has a $(1,j)$-set $D$ of cardinality at most $k=t+q+3jq^2$. 
First observe that since $D$ is a dominating set, $D$ must contain at least $t$ vertices from the 
set 
$$
  V_{4} \stackrel{\rm def}{=} \{ y_{1}, \, \dots, \, y_{t}\} \cup \{ z_{1}, \, \dots, \, z_{t}\} \cup \{ u_{1}, \, \dots, \, u_{t}\}
$$
to dominate the pendant vertices 
$$
  \{ y_{1}, \, \dots, \, y_{t}\} \cup \{ z_{1}, \, \dots, \, z_{t}\}.
$$
Similarly, for a fixed $i$ and $r$, consider the tree $T^{i}_{r}$. To dominate the 
pendant vertices of the tree $T^{i}_{r} = (V^{i}_{r}, E^{i}_{r})$ we need to select at least $j$ vertices from the set 
$V^{i}_{r} \setminus \{w^{i}_{r}\}$.
Summing up over all $i$ and $r$, we get that $D$ contains more than $3jq^{2}$ vertices from the set 
$$
  V_{5} \stackrel{\rm def}{=} \bigcup_{1\leq i \leq 3q, \, 1\leq r \leq q} \left( V^{i}_{r} \setminus \{w^{i}_{r}\} \right).
$$
Observe now that the cardinality of $D$ is at least $t+3q.jq=t+3jq^2$.

Now to complete the proof we will only have to show that $V_{2} \cap D = \emptyset$.
Since if this is the case then each $x_i$ has to be dominated by either
some $w^{i}_{l} \in G_i$ or a $v_s \in V_1$, $s\in [t]$. 
We have to dominate the $3q$ vertices of $V_2$ using at most $q$ vertices,
since we have used up the other $t+3jq^2$ vertices. Since each $w^{i}_{l}$ dominates only one $x_i$, while each $v_i \in V_1$
dominates $3$ $x_i$'s, this is possible only if there exist $q$ vertices $v_{i_1},\ldots,v_{i_q}$, which can dominate the $3q$ vertices $x_i\in V_2$.
Now define $C'$ to be the sets corresponding to these vertices, i.e. $C'= \{C_{i_1},\, \ldots,\, C_{i_q}\}$. Clearly $C'$ is an exact cover of $X$, and has
only $q$ sets.

Till now we have only used the fact that $D$ is a dominating set but for showing $D \cap V_{2} = \emptyset$
we will be crucially using the fact that $D$ is a $(1,j)$-set.
To reach a contradiction let us
suppose some $x_i \in D$. Then each $w^{i}_{r}$, $r\in [q]$ is $1$-dominated by $x_i$,
and either has to be in $D$ or can have at most $j-1$ other neighbours that are in $D$.
In either case, we get that for each tree $T^{i}_{r} \in G_i$, $|T^{i}_{r} \cap D| \geq j+1$.
Hence, $|G_i \cap D| \geq jq +q$. This implies that $|D|\geq t+1+(3q-1)(jq)+(j+1)q = t+1+3jq^2+q$,
which contradicts the assumption that $|D|\leq k$. 
Therefore $D\cap V_{2} = \emptyset$.
\end{proof}



\begin{remark}
Following observations directly
follow from the NP hardness reduction:
\begin{enumerate}
 \item 
    The only possibility of dominating $V_4$ by $t$ vertices is to take $\{u_1,\, u_2,\, \dots,\, u_t\}$ and 
    this set also dominates the set $\{v_1,\, v_2,\, \dots,\, v_t\}$.

 \item
    The only possibility of dominating $V_5$ by $3jq^2$ vertices is to take $\bigcup_{i=1}^{3q} S_i$ 
    and this set dominates each $w^{i}_{l}$ exactly $j$ times. Note that 
    $S_i$ is the set of all children of $w_1^i,\, w_2^i,\, \ldots,\, w_{q}^i$ and
    each $S_i$ contains $jq$ vertices.
\end{enumerate}
\end{remark}

\section{Polynomial time algorithms}
\label{sec-polynomial-time-algorithms}

\subsection{Tree}
\label{ssec-tree}

To design an efficient algorithm for finding $(1,j)$-domination number of a given tree $T$, we need the concept 
of a more generalized set, namely $M$-set of an $M$-labeled tree. In fact, we design a dynamic programming 
algorithm for finding the minimum cardinality of an $M$-set of an $M$-labeled tree. First let us define 
an $M$-labeled tree and an $M$-set.

\begin{defi}
A tree $T$ is called an \emph{$M$-labeled tree} if each vertex $v$ is associated with two nonnegative integers 
$M_a(v)$ and $M_b(v)$ such that $M_a(v)\leq M_b(v)$. A subset $S\subseteq V$ of an $M$-labeled tree $T=(V,E)$ 
is called an \emph{$M$-set} if $M_a(v)\leq |N_T(v)\cap S|\leq M_b(v)$ for every $v\in V\setminus S$. The minimum 
cardinality of an $M$-set of an $M$-labeled tree $T$ is called the \emph{$M$-domination number} of 
$T$ and is denoted by $\gamma_M(T)$.
\end{defi}

Note that if all the vertices of an $M$-labeled tree $T$ can be labeled as $M_a(v)=1$ and $M_b(v)=j$, 
then an $M$-set of $T$ is nothing but a $(1,j)$-set of the underlying tree.

The main idea of the dynamic programming algorithm is to choose a specific vertex $u$ from $T$. Any minimum 
$M$-set of $T$ should either contain $u$ or does not contain $u$. So the problem of finding the minimum 
cardinality of an $M$-set of $T$ boils down to finding two parameters: $(i)$ $\gamma_M(T, u)$, the minimum 
cardinality of an $M$-set of $T$ that contains the specific vertex $u$ and $(ii)$ $\gamma_M(T, \bar{u})$, 
the minimum cardinality of an $M$-set of $T$ that does not contain the specific vertex $u$.

Suppose $uv$ is an edge of the $M$-labeled tree $T$. Let $H_1$ and $H_2$ be the subtrees of $T$ rooted at 
$u$ and $v$ respectively. Note that $H_1$ and $H_2$ are $M$-labeled trees and the labels of the vertices of 
$H_1$ and $H_2$ remain the same as they are in $T$. Our aim is to use the parameters 
$\gamma_M(H_1, u)$, $\gamma_M(H_1, \bar{u})$, $\gamma_M(H_2, v)$, and $\gamma_M(H_2, \bar{v})$ 
(with suitable labeling $M$) to find $\gamma_M(T, u)$ and $\gamma_M(T, \bar{u})$. The following lemma 
shows how the values of $\gamma_M(T, u)$ and $\gamma_M(T, \bar{u})$ are obtained.

\begin{lemma}\label{lemtreealgo}
Let $uv$ be an edge of an $M$-labeled tree $T$ and $H_1$ and $H_2$ be the subtrees of 
$T$ rooted at $u$ and $v$ respectively. Then the following statements hold.
\begin{itemize}
\item[(a)] $\gamma_M(T,u)= \gamma_M(H_1,u)+ \gamma_{M'}(H_2)$, where the label $M'$ 
is same as $M$ except $M'_a(v)=\max \{M_a(v)-1, 0\}$ and $M'_b(v)=\max \{M_b(v)-1, 0\}$

\item[(b)] $\gamma_M(T, \bar{u})= \min \{ \gamma_M(H_1, \bar{u})+ \gamma_M(H_2, \bar{v}), \gamma_{M'}(H_1, \bar{u})+ \gamma_M(H_2, v) \}$, 
where the label $M'$ is same as $M$ except $M'_a(u)=\max \{M_a(u)-1, 0\}$ and $M'_b(u)=\max \{M_b(u)-1, 0\}$.

\end{itemize}
\end{lemma}

\begin{proof}
(a) Let $D$ be a minimum cardinality $M$-set of $T$ containing $u$. Let $D_1= V(H_1)\cap D$ and 
$D_2= V(H_2)\cap D$. Clearly $D_1$ is an $M$-set of $H_1$ containing $u$. Now, $D_2$ may or may 
not contain the vertex $v$.
\begin{description}
\item[Case $v\in D_2$:] In this case, $D_2$ is an $M'$-set of $H_2$ containing $v$.

\item[Case $v\notin D_2$:] In this case, $D_2$ is an $M'$-set of $H_2$ not containing $v$.
\end{description}
Since $\gamma_{M'}(H_2)=\min \{\gamma_{M'}(H_2, v), \gamma_{M'}(H_2, \bar{v})\}$, 
we have $\gamma_M(H_1,u)+ \gamma_{M'}(H_2) \leq \gamma_M(T,u)$.

On the other hand, let $D_1$ be a minimum cardinality $M$-set of $H_1$ containing $u$ and $D_2$ be a 
minimum cardinality $M'$-set of $H_2$. Let $D=D_1\cup D_2$. Clearly, whatever be the case 
($v\in D_2$ or $v\notin D_2$), we can verify that $D$ is a $M$-set of $T$ and $u\in D$. Hence, 
$\gamma_M(T,u)\leq \gamma_M(H_1,u)+ \gamma_{M'}(H_2)$.

Thus we have, $\gamma_M(T,u)= \gamma_M(H_1,u)+ \gamma_{M'}(H_2)$.

(b) Let $D$ be a minimum cardinality $M$-set of $T$ not containing $u$. Let 
$D_1= V(H_1)\cap D$ and $D_2= V(H_2)\cap D$. Now, $D_2$ may or may not contain the vertex $v$.
\begin{description}
\item[Case $v\notin D$:] In this case, $D_1$  is an $M$-set of $H_1$ not containing $u$ and $D_2$ is an 
$M$-set of $H_2$ not containing $v$. Hence, $\gamma_M(H_1, \bar{u})+ \gamma_M(H_2, \bar{v})\leq \gamma_M(T, \bar{u})$.

\item[Case $v\in D$:] In this case, $D_1$  is an $M'$-set of $H_1$ not containing $u$ and $D_2$ 
is an $M$-set of $H_2$ containing $v$. Hence, $\gamma_{M'}(H_1, \bar{u})+ \gamma_M(H_2, v)\leq \gamma_M(T, \bar{u})$.
\end{description}
So we have, $\min \{ \gamma_M(H_1, \bar{u})+ \gamma_M(H_2, \bar{v}), \gamma_{M'}(H_1, \bar{u})+ \gamma_M(H_2, v) \} \leq \gamma_M(T, \bar{u})$.

On the other hand, for showing 
$\gamma_M(T, \bar{u})\leq \min \{ \gamma_M(H_1, \bar{u})+ \gamma_M(H_2, \bar{v}), \gamma_{M'}(H_1, \bar{u})+ \gamma_M(H_2, v) \}$, 
we have the following two cases:

\begin{description}
\item[Case $\min \{ \gamma_M(H_1, \bar{u})+ \gamma_M(H_2, \bar{v}), \gamma_{M'}(H_1, \bar{u})+ \gamma_M(H_2, v) \}= \gamma_M(H_1, \bar{u})+ \gamma_M(H_2, \bar{v})$:]

Let $D_1$ be a minimum cardinality $M$-set of $H_1$ not containing $u$ and $D_2$ be a 
minimum cardinality $M$-set of $H_2$ not containing $v$. Let $D=D_1\cup D_2$. We can easily 
verify that $D$ a minimum cardinality $M$-set of $T$ not containing $u$. So, 
$\gamma_M(T, \bar{u})\leq \gamma_M(H_1, \bar{u})+ \gamma_M(H_2, \bar{v})$.

\item[Case $\min \{ \gamma_M(H_1, \bar{u})+ \gamma_M(H_2, \bar{v}), \gamma_{M'}(H_1, \bar{u})+ \gamma_M(H_2, v) \}=\gamma_{M'}(H_1, \bar{u})+ \gamma_M(H_2, v)$:]

In this case, similarly we can show that $\gamma_M(T, \bar{u})\leq \gamma_{M'}(H_1, \bar{u})+ \gamma_M(H_2, v)$.
\end{description}

Hence in both the cases, $\gamma_M(T, \bar{u})\leq \min \{ \gamma_M(H_1, \bar{u})+ \gamma_M(H_2, \bar{v}), \gamma_{M'}(H_1, \bar{u})+ \gamma_M(H_2, v) \}$.

Thus we have, $\gamma_M(T, \bar{u})= \min \{ \gamma_M(H_1, \bar{u})+ \gamma_M(H_2, \bar{v}), \gamma_{M'}(H_1, \bar{u})+ \gamma_M(H_2, v) \}$.
\end{proof}

Based on the above lemma, we have the following dynamic programming algorithm for finding 
$\gamma_M(T)$ for an $M$-labeled tree $T$. Note that, a tree with a single vertex forms the base case at 
which $\gamma_M(T)$ can be easily computed depending upon the $M$ label.

\begin{algorithm}\label{algotree}
\KwIn{A $M$-labeled tree $T=(V,E)$.}
\KwOut{A minimum cardinality of an $M$-set of $T$, i.e., $\gamma_M(T)$.}
\Begin{
Select a vertex $u$ from $V$;\\
Select an edge $uv$ from $E$;\\
Calculate $\gamma_M(T, u)$ and $\gamma_M(T, \bar{u})$ according to Lemma~\ref{lemtreealgo};\\
$\gamma_M(T)= \min\{\gamma_M(T, u), \gamma_M(T, \bar{u})\}$;\\
\Return $\gamma_M(T)$;

}

\caption{Min\_M-set\_Tree}
\end{algorithm}

The correctness of Algorithm \ref{algotree} is based on Lemma~\ref{lemtreealgo}. 
Since the dynamic programming runs over the edges of the given tree, Algorithm~\ref{algotree} take linear time. 
Also, as noted earlier, if we initialize the $M$-label as $M_a(v)=1$ and $M_b(v)=j$ for all $v\in V$, 
then $\gamma_{(1,j)}(T)=\gamma_M(T)$. Hence we have the following theorem.

\begin{theo}
The $(1,j)$-domination number of a given tree can be computed in linear time.
\end{theo}

\subsection{Split graph}

In this subsection, we design an algorithm which finds $(1,j)$-domination number for a given split graph $G$ 
in polynomial time. This algorithm is important because most of the domination type problems like 
domination~\cite{domsplit}, total domination~\cite{tdomsplit}, $k$-tuple domination~\cite{kdomsplit} etc.  
are NP-complete for split graphs.

Let the vertex set $V$ of a split graph $G=(V,E)$ is partitioned into a clique $K$ and an 
independent set $S$, i.e., 
$V=K\cup S$. Also assume that $|K|=n_1$ and $|S|=n_2$. Note that in finding a minimum $(1,j)$-set, $j$ can be 
considered as a constant. Now if $n_1\leq j$, then we are done. Because, in that case, we consider all possible 
subsets of $K$ and based on the neighborhood set of these subsets we can find a minimum cardinality $(1,j)$-set. 
Since $j$ is a constant, the number of subsets of $K$ is bounded by a constant (this constant is huge, $2^j$). 
This implies that, in this case, we can find a minimum $(1,j)$-set in polynomial time. Hence we assume that $j< n_1$. 
The idea of the algorithm is based on a simple fact that if a $(1,j)$-set, say $D$, contains more than $j$ but 
less than $n_1$ vertices from $K$, then there exists a vertex in $K\setminus D$ which is dominated by more than 
$j$ vertices, which is a contradiction to the definition of $(1,j)$-set. Hence we have the following observation.

\begin{obs}\label{obsalgosplit}
Every $(1,j)$-set of a given split graph $G$ contains only $i$ vertices from $K$ where $i\in \{0,1,2,\ldots, j, n_1\}$.
\end{obs}

Now, for each $i=0, 1, 2,\ldots, j$ and $n_1$, we find a minimum cardinality $(1,j)$-set $D$ of $G$ such that 
$|K\cap D|=i$. Finally we pick the minimum cardinality $(1,j)$-set among these $j+2$ types of $(1,j)$-sets of $G$. 
Hence the main task in this algorithm is to find a minimum cardinality $(1,j)$-set $D$ of $G$ such that 
$|K\cap D|=i$ for each $i=0, 1, 2,\ldots,j$ and $n_1$. The following lemma gives a complete characterization 
of these $j+2$ types of $(1,j)$-sets.

\begin{lemma} \label{lemsplitalgo}
Let the vertex set $V$ of a connected split graph $G=(V,E)$ is partitioned into a clique $K$ and an independent 
set $S$, i.e., $V=K\cup S$ and $|K|=n_1$ and $|S|=n_2$. Let $D$ be a minimum $(1,j)$-set of $G$. Then the 
following statements are true.
\begin{itemize}
\item[(a)] If $K\cap D=\emptyset$, then $d_G(v)\in \{n_1, n_1+1, \ldots, n_1+j-1\}$ for all $v\in K$. 
In this case, $D=S$ is the only $(1,j)$-set of $G$.

\item[(b)] For all $i\in [j-1]$, if $K\cap D=\{v_1, \ldots, v_i\}=K_i$, then 
$d_{G[K\cup S_i]}(v) \in \{n_1-1, n_1, \ldots, n_1+(j-i)-1\}$ for all $v\in K\setminus K_i$, where $S_i= S\setminus N_G(K_i)$. 
In this case, $D=K_i\cup S_i$ is a minimum cardinality $(1,j)$-set of $G$ containing $K_i$.

\item[(c)] If $K\cap D=\{v_1, v_2, \ldots, v_j\}=K_j$, then $S\subset N_G(K_j)$. In this case, 
$D=\{v_1, v_2,\ldots, v_j\}$ is a $(1,j)$-set of $G$ of minimum cardinality.

\item[(d)] If $K\subseteq D$, then $S_2\subseteq D$, where $S_2=\{u\in S|~d_G(u)\geq (j+1)\}$. In this case, 
$D=K\cup S_2$ is a $(1,j)$-set of $G$ of minimum cardinality.
\end{itemize}
\end{lemma}

\begin{proof}
(a) In this case, since $S$ is an independent set, $D=S$. Again, since $K$ is a clique, $d_G(v)\geq n_1-1$ 
for all $v\in K$. If $d_G(v_i)= n_1-1$ for some vertex $v_i\in K$, then 
$$
	N_G(v_i)\cap D = \emptyset. 
$$
This 
is a contradiction to the definition of $(1,j)$-set. Again if $d_G(v_t)\geq n_1+j$ for some vertex $v_t\in K$, 
then 
$$
	|N_G(v_t)\cap D|\geq (j+1). 
$$
This will force $v_t$ to be in $D$, which is a contradiction to 
$K\cap D=\emptyset$. Thus, we have 
$$
	d_G(v)\in \{n_1,\, n_1+1,\, \ldots,\, n_1+j-1 \} \quad  \forall v \in K.
$$

(b) In this case, clearly $(K_i\cup S_i)\subseteq D$. Again, since $K$ is a clique, $d_{G[K\cup S_i]}(v)\geq n_1-1$ 
for all $v\in K$. If $d_{G[K\cup S_i]}(v_t)\geq n_1+(j-i)$ for some vertex $v_t\in (K\setminus K_i)$, then 
$|N_G(v_t)\cap D|\geq j$. This will force $v_t$ to be in $D$, which is a contradiction to $|K\cap D|=i$. Thus 
$$
	d_{G[K\cup S_1]}(v) \in \{n_1-1, \, n_1,\, \ldots,\, n_1+(j-i)-1\} \quad \forall v\in K\setminus \{v_1\}. 
$$
In this case, 
clearly $D=K_i\cup S_i$ is a minimum cardinality $(1,j)$-set of $G$ containing $K_i$.

(c) If possible, let $u\in S\setminus N_G(K_j)$. Clearly $u\in D$. Since $G$ is connected, there exists at least 
one vertex $v$ in $N_G(u)\setminus K_j$. Now for that vertex $v$, 
$$
	|N_G(v)\cap D|\geq j.
$$ 
This will force $v$ to 
be in $D$, which is a contradiction to $|K\cap D|=j$. Thus $S\subset N_G(K_j)$. Clearly in this case, $D=K_j$ is 
a $(1,j)$-set of $G$ of minimum cardinality.

(d) The proof is trivial and hence omitted.
\end{proof}

Based on the above lemma, we have Algorithm \ref{algosplit} that finds a minimum cardinality $(1,j)$-set 
of a given split graph $G$.

\begin{algorithm}[h]\label{algosplit}
\relsize{-1}{\KwIn{A split graph $G=(V,E)$ with $V=K\cup S$.}
\KwOut{A minimum cardinality $(1,j)$-set $D$ of $G$.}

\Begin{

\If{$d_G(v)\in \{n_1, n_1+1, \ldots, n_1+j-1\}$ for all $v\in K$}{
$D_0==S$;}
\Else
{$D_0==\emptyset$;}

\ForEach{$i\in [j-1]$}{
\ForEach{$i$ element subsets $K^i_t$ of $K$}{
\If{$d_{G[K\cup S^i_t]}(v) \in \{n_1-1, n_1,\ldots, n_1+(j-i)-1\}$ for all 
$v\in K\setminus K^i_t$, where $S^i_t= S\setminus N_G(K^i_t)$}{
$D^i_t==K^i_t\cup S^i_t$;}
\Else
{$D^i_t== \emptyset$;}
}
$D_i==D'$, where $D'$ be the minimum nonempty set among all $D^i_t$, for all $t$;
}

\ForEach{$j$ element subsets $K^j_t$ of $K$}{
\If{$S\subset N_G(K^j_t)$}{
$D_j==K^j_t;$
}
\Else
{$D_j==\emptyset$}
}

$D_{j+1}==K\cup S'$, where $S'=\{u\in S|~d_G(u)\geq (j+1)\}$;

$D==D_p$, where $D_p$ is the minimum cardinality nonempty set among all $D_i$, $0\leq i\leq (j+1)$.

\Return $D$;\\
}

\caption{Min\_(1,j)-set\_Split} }
\end{algorithm}

The correctness of Algorithm~\ref{algosplit} is based on Observation~\ref{obsalgosplit} and 
Lemma~\ref{lemsplitalgo}. Next we analyze the complexity of Algorithm~\ref{algosplit}. Note that we can 
compute $D_0$ and $D_{j+1}$ in $O(n)$ time. For each $1\leq i\leq (j-1)$, the set $D_i$ can be computed in 
$O(n^i+ i n^i \log n)$ time because in each case, we have to check all possible $i$ element subsets of $K$, 
i.e., $O(n^i)$ subsets and after that we have to assign the minimum cardinality subset to $D_i$, which takes 
$O(i n^i \log n)$ time. Hence computing all the sets $D_i$ for $0\leq i\leq (j+1)$ can be done in polynomial 
time. For $D_j$, we have to check all $j$ element subsets of $K$, i.e., $O(n^j)$, and since $j$ is a constant, 
it takes polynomial time to compute $D_j$. Also finding the minimum cardinality nonempty set in line $19$ takes 
polynomial time. Hence, Algorithm~\ref{algosplit} can be done polynomial time. Thus, we have the main theorem 
in this subsection as follows.

\begin{theo}
For any fixed $j$, the cardinality of a minimum $(1,j)$-set of a given split graph can be computed in polynomial time.
\end{theo}

\section{Concluding remarks}

In this paper, we have obtained an upper bound on $(1,j)$-domination number. We have shown that $(1,j)$-SET 
is NP-complete for chordal graphs. We have also designed two algorithms for finding $\gamma_{(1,j)}(G)$ of a 
tree and a split graph. In~\cite{12set}, the authors constructed a special type of split graph $G$ with $n$ 
vertices for which $\gamma_{(1,2)}(G)=n$. Lemma~\ref{lemsplitalgo} gives a more general type of split graph 
for which $\gamma_{(1,j)}(G)=n$. It actually characterizes the split graphs with $n$ vertices having 
$\gamma_{(1,j)}(G)=n$. The characterization is as follows:
\begin{coro}
Let $G$ be a split graph with $V=K\cup S$ and $|K|=n_1$, $|S|=n_2$. Then $\gamma_{(1,j)}(G)=n$ if and 
only if the following conditions hold.
\begin{itemize}
\item[(i)] There exists at least one vertex $v$ in $K$ such that $d_G(v)\geq n_1+j$.
\item[(ii)] For all $i\in [j-1]$ and for each $i$ element subset $K_i=\{v_1, \ldots, v_i\}$ of $K$, 
there exists some vertex $v_t\in K\setminus K_i$ such that $d_{G[K\cup S_i]}(v_t)\geq n_1+(j-i)$, where $S_i= S\setminus N_G(K_i)$.
\item[(iii)] For each $j$ element subset $K_j$ of $K$, $S\not\subset N_G(K_j)$.
\item[(iv)] For every $u\in S$, $d_G(u)\geq (j+1)$.
\end{itemize}
\end{coro}

Condition $(i), (ii)$ and $(iii)$ actually force all the vertices of $K$ in a minimum cardinality $(1,j)$-set 
$D$ and condition $(iv)$ forces all vertices of $S$ in $D$. Using this type of split graphs, we can construct 
a graph $G$ (not a split graph) having $\gamma_{(1,j)}(G)=n$. The construction is as follows: Let 
$$
	G_1=(V_1, E_1),\, G_2=(V_2, E_2),\, \ldots,\, G_p=(V_p, E_p)
$$ 
be $p$ split graphs having partitions 
$$
	V_1=K_1\cup S_1,\, V_2=K_2\cup S_2,\, \ldots,\, V_p=K_p\cup S_p. 
$$
The vertex set of the constructed graph $G=(V,E)$ 
is given by 
$$
	V=\bigcup_{i=1}^p V_i
$$ 
and the edge set is given by 
$$
	E = \left( \bigcup_{i=1}^p E_i \right) \bigcup E',
$$ 
where $E'$ is an 
arbitrary edge set between the vertices of 
$$
	K_1,\, K_2,\, \ldots,\, K_p.
$$ 
We can easily verify that $\gamma_{(1,j)}(G)= |V|$. 
But characterizing the graphs with $n$ vertices having $\gamma_{(1,j)}(G)=n$ 
seems to be an interesting but difficult question. 
Also, it would be interesting to study the open problems mentioned in~\cite{12set}.

\subsection*{Acknowledgements}

Kunal Dutta and Arijit Ghosh are 
supported by the Indo-German Max Planck Center for Computer Science (IMPECS).
Subhabrata Paul is supported by the Indian Statistical Institute, Kolkata.

\bibliographystyle{alpha}
\addcontentsline{toc}{section}{Bibliography}
\bibliography{1-j}

\end{document}